\documentclass[pra,superscriptaddress,showpacs,twocolumn]{revtex4}

\usepackage{amsfonts,amssymb,amsmath}
\usepackage{graphics,graphicx,epsfig}
\usepackage{amsthm}

\usepackage{hyperref}

\def\identity{\leavevmode\hbox{\small1\kern-3.8pt\normalsize1}}

\newtheorem{theorem}{Theorem}

\newtheorem{lemma}{Lemma}

\newcommand{\ket}[1]{\left | #1 \right\rangle}
\newcommand{\bra}[1]{\left \langle #1 \right |}

\newcommand{\Tr}{\mathrm{Tr}}

\newcommand{\proj}[1]{\ket{#1}\bra{#1}}
\renewcommand{\epsilon}{\varepsilon}

\begin{document}

\title{Correlations between outcomes of random observables}

\author{Minh Cong Tran}
\affiliation{School of Physical and Mathematical Sciences, Nanyang Technological University, Singapore}

\author{Borivoje Daki\'c}
\affiliation{Faculty of Physics, University of Vienna, Boltzmanngasse 5, A-1090 Vienna, Austria}
\affiliation{Institute for Quantum Optics and Quantum Information, Austrian Academy of Sciences, Boltzmanngasse 3, A-1090 Vienna, Austria}

\author{Wies{\l}aw Laskowski}
\affiliation{Institute of Theoretical Physics and Astrophysics, University of Gda\'nsk, Gda\'nsk, Poland}

\author{Tomasz Paterek}
\affiliation{School of Physical and Mathematical Sciences, Nanyang Technological University, Singapore}
\affiliation{Centre for Quantum Technologies, National University of Singapore, Singapore}
\affiliation{MajuLab, CNRS-UNS-NUS-NTU International Joint Research Unit, UMI 3654, Singapore}

\begin{abstract}
We recently showed that multipartite correlations between outcomes of random observables detect quantum entanglement in all pure and some mixed states.
In this follow up article we further develop this approach, derive maximal amount of such correlations and show that they are not monotonous under local operations and classical communication.
Nevertheless, we demonstrate their usefulness in entanglement detection with a single random observable per party.
Finally we study convex-roof extension of the correlations and provide a closed-form necessary and sufficient condition for entanglement in rank-$2$ mixed states and a witness in general.
\end{abstract}


\maketitle

The Bell singlet state is a paradigmatic example of an entangled state.
This is usually demonstrated by noting that the entropy of the pair of particles is smaller than the entropy of each particle, a possibility forbidden in classical objects.
At the same time the singlet state is famous for its correlations.
Indeed, two observers measuring the same spin direction will always find their outcomes opposite.
This holds independently of a particular measurement direction, in agreement with the total spin being zero.
Furthermore, even for spin directions that differ quantum mechanics predicts high probability of opposite outcomes.
One might therefore ask if correlations between randomly chosen observables reveal entanglement.
We have recently shown that such ``random correlations'' are indeed the feature of entangled pure states~\cite{PhysRevA.92.050301}:
A pure $N$-particle state is entangled if and only if the squared $N$-partite correlation functions averaged over uniform choices of local observables exceed certain bound.

In this follow up paper we extend our approach in several ways.
In Sec.~\ref{SEC_PURE} we focus on pure states in arbitrary dimensions and derive explicitly equivalence between entanglement and the random correlations in general.
We then study maximal amount of random correlations in a pure state and find that it is achieved (non uniquely) by the Greenberger-Horne-Zeilinger (GHZ) states of odd number of qubits.
(We conjecture that GHZ states give rise to maximal random correlations in general).
It turns out that random correlations of the 2d cluster states scale intermediately as expected from entanglement of resources for universal quantum computing~\cite{PhysRevLett.102.190501,PhysRevLett.102.190502}.
All this suggests that random correlations might be a proper entanglement monotone.
We show that this is true for bipartite systems and provide explicit counter-examples for a five-qubit system.
Nevertheless, the random correlations are helpful as entanglement witnesses which we demonstrate on a vivid example where entanglement is detected with \emph{one} random observable per party.

In Sec.~\ref{SEC_MIXED} we move to mixed states and consider convex roof extension of random correlations.
We prove necessary and sufficient condition for entanglement in rank-$2$ states and present entanglement witness for general states.
The witness is illustrated on an explicit example where it detects all entangled states of certain family.

\section{Pure states}
\label{SEC_PURE}

Let us briefly summarize previous results regarding random correlations and a related notion of the length of correlations
(not to be confused with the length in physical space).
We start with two-level systems, qubits.
Any $N$-qubit density matrix can be represented in terms of Pauli matrices as
\begin{align}
	\rho = \frac{1}{2^N}\sum_{\mu_1,\dots,\mu_N=0}^3 T_{\mu_1\dots\mu_N}\sigma_{\mu_1}\otimes\dots\otimes\sigma_{\mu_N},
\end{align}
where $\sigma_1, \sigma_2,\sigma_3$ are the Pauli matrices and $\sigma_0$ is the identity. 
The real coefficients $ T_{\mu_1\dots\mu_N}\in [-1,1]$ form an extended correlation tensor
which is just an alternative to the density matrix representation of a quantum state.
If each party chooses to measure their qubit along a local direction $\vec u_i$ then expectation of the product of their measurement outcomes, so called correlation function, is given by 
\begin{align}
	E(\vec u_1, \dots , \vec u_N) = \sum_{j_1, \dots, j_N = 1}^3 T_{j_1\dots j_N}  \: (\vec u_1)_{j_1} \dots (\vec u_N)_{j_N},
\label{E_ROT_INV}
\end{align}
where $(\vec u_i)_{j_i}$ is the $j_i$th component of vector $\vec u_i$.
We define \emph{random correlations} as the expectation value of squared correlation functions averaged over uniform choices of settings for each individual observer
\begin{align}
	\mathcal{R} \equiv \frac{1}{(4\pi)^N} \int d \vec u_1 \dots \int d \vec u_N \, \,  E^2(\vec u_1, \dots, \vec u_N), \label{DEF_TINF}
\end{align}
where $d \vec u_n = \sin \theta_n d \theta_n d \phi_n$ is the usual measure on the unit sphere.
To estimate $\mathcal{R}$, it would seem that we have to take into account all local directions but in fact it is sufficient to consider only a set of orthogonal axes for each party \cite{PhysRevA.92.050301}:
\begin{align}
	\mathcal{R}=\frac{1}{3^N} \sum_{\vec u_1,\dots,\vec u_N = \vec{x},\vec{y},\vec{z}} E^2(\vec u_1,\dots,\vec u_N)\equiv\frac{ \mathcal{C}}{3^N}, \label{DEF_C}
\end{align}
where  \emph{length of correlations} $\mathcal{C}$ is defined as the sum of squared correlations measured along a complete set of orthogonal local axes.
Note that no common reference frame is required, each observer is allowed a different Cartesian coordinate system.
The name ``length of correlations'' refers to the fact that $\mathcal{C}$ is a squared norm of the correlation tensor (with components having solely correlation functions between all $N$ qubits).
It has been shown in Refs.~\cite{QuantInfComp.8.773,PhysRevA.77.062334,PhysRevA.80.042302,PhysRevA.92.050301} that $\mathcal{C}>1$ if and only if the system is in a pure entangled state.
In Appendix~\ref{APP_SCHMIDT} we present a simple alternative proof for pure systems of two and three qubits that follows directly from the Schmidt decomposition.
This line of reasoning does not extend to higher number of qubits because the restrictions brought forward by the Schmidt decomposition do not engage sufficient number of subsystems.


\subsection{Higher dimensions}

We now extend our criteria for entanglement to higher dimensions, i.e. qudits of dimension $d$. 
The final result is already presented in Ref.~\cite{PhysRevA.92.050301}, but we would like to clarify which quantities exactly we consider and discuss explicitly some subtleties.
The step-by-step derivation below is presented for the first time.

We replace the Pauli matrices with a complete orthonormal basis consisting of identity and $d^2-1$ traceless operators $\sigma_j$ such that for all $j,k = 1,\dots,d^2-1$:
\begin{align}
	&\Tr(\sigma_j) = 0,	\label{COND_BASIS1}\\	
	&\Tr(\sigma_j \sigma^\dag_k) = d \, \delta_{jk}. \label{COND_BASIS2}
\end{align}
Various concrete realisations can be taken here, e.g. generalised Gell-Mann operators (hermitian basis) or the Weyl-Heisenberg operators (unitary basis).
For completeness we write them in Appendix~\ref{APP_BASES}.

An arbitatry state $\rho$ of $N$ qudits can be decomposed using these operators in a way similar to the Bloch representation:
\begin{align}
	\rho=\frac{1}{d^N}\sum_{\mu_1,\mu_2,\dots,\mu_N=0}^{d^2-1}T_{\mu_1\mu_2\dots\mu_N}\sigma_{\mu_1}\otimes\dots\otimes\sigma_{\mu_N},
\end{align}
with $\sigma_0 = \mathbb{I}$ being identity and $\sigma_j$ the operators defined above. 
As before, the coefficients are:
\begin{align}
	 T_{\mu_1\mu_2\dots\mu_N}=\Tr{(\rho \, \sigma_{\mu_1}^\dag\otimes\dots\otimes\sigma_{\mu_N}^\dag}). \label{DEF_CORRd}
\end{align}
However, these coefficients are in general complex valued as $\sigma$'s are not required to be hermitian.
We therefore define the \emph{length of correlations} with the help of absolute value:
\begin{align}
	\mathcal{C} \equiv \sum_{j_1,j_2,\dots,j_N=1}^{d^2-1}|T_{j_1j_2\dots j_N}|^2.
	\label{C_D}
\end{align}
As proven in Appendix \ref{APP_CInv}, the length of correlations $\mathcal{C}$ is invariant under the choice of local basis as long as \eqref{COND_BASIS1} and \eqref{COND_BASIS2} are satisfied.
Note that this is more general than invariance under local unitary transformations.
For example, the Gell-Mann basis cannot be obtained from the Weyl-Heisenberg basis by local unitaries because the corresponding operators have different eigenvalues.
Nevertheless, the length of correlations is the same in both bases provided they are suitably normalized.

Similar to the case of qubits, the length of correlations can also be used to identify entanglement in pure states.
\begin{theorem} 
\label{TH_1}
Pure state $\ket{\Psi}$ is entangled if and only if
\begin{align}
	\mathcal{C} > (d-1)^N.
\end{align}
\end{theorem}
\begin{proof}
The proof follows the same lines as for qubits and is presented in Appendix \ref{APP_TH_1}.
\end{proof}
It should also be clear how to extend this theorem to cover subsystems of different dimensions, e.g. $2 \times 3$.

We will now define random correlations for qudits and show that they are proportional to $\mathcal{C}$.
The generalization is obtained by requesting that random correlations are the average of squared correlations over uniform choices of unitary matrices $U_n$ for each observer:
\begin{align}
	\mathcal{R} \equiv \frac{1}{W^N} \int d U_1 \dots \int d U_N \, \,  |E(U_1, \dots, U_N)|^2, \label{DEF_Rd}
\end{align}
where $W$ is a normalization constant, i.e. $\int d U_n = W$, and the unitary-dependent correlation function reads:
\begin{align}
	E(U_1,\dots,U_N) = \Tr\left[\rho \, \bigotimes_{n=1}^N U_n^\dag \sigma_1^\dag U_n \right].
\end{align}
Here we have chosen exemplary operator $\sigma_1$ as an initial observable of the averaging.
In Appendix~\ref{APP_INITIAL} we prove that $\mathcal{R}$ is independent of the choice of this initial operator as long as it is traceless and normalized.
With these definitions we link random correlations and entanglement as follows.
\begin{theorem}
For any pure state of $N$ qudits, each of dimension $d$,
\begin{align}
	\mathcal{R} = \frac{\mathcal{C}}{(d^2-1)^N}.
\end{align}
\begin{proof}
Note that the length of correlations $\mathcal{C}$ is a sum of $(d^2-1)^N$ squared correlation functions, each of which is equal to $\mathcal{R}$ after averaging over local unitary transformations (see Appendix~\ref{APP_INITIAL}). 
Since $\mathcal{C}$ is invariant under such transformations, overall these averages still sum up to $\mathcal{C}$.
Therefore $\mathcal{C}$ is just a multiple of $\mathcal{R}$.
\end{proof}
\end{theorem}


\subsection{Maximal random correlations}

Here we examine states that maximize $\mathcal{C}$ (and therefore $\mathcal{R}$).
We focus on multiple qubits.
Let us begin by proving the upper bound on the length of correlations.
\begin{theorem}\label{TH_MAXENT}
Every pure state of odd number of qubits $N$ satisfies
\begin{align}
	\mathcal{C}\leq 2^{N-1}. \qquad \qquad (N \textrm{ is odd})
	\label{EQN_MAXCORR}
\end{align} 
\end{theorem}
\begin{proof}
Let us denote by $C_k$ the sum of squared correlations between any $k$ observers. 
In particular, for a system of $N$ qubits in a state $\ket{\psi}$ we have $C_N = \mathcal{C}(\psi)$. 
The purity condition of $\ket{\psi}$ implies
\begin{align}
	1 + C_1 + C_2 +C_3+ \dots + C_{N-1}+C_N = 2^N, \label{MaxCorrProof1}
\end{align}
whereas Ref.~\cite{ES15} demonstrates for $N$ odd
\begin{align}
	1-C_1 + C_2 - C_3 +\dots + C_{N-1} - C_N = 0. \label{MaxCorrProof2}
\end{align}
Summing up \eqref{MaxCorrProof1} and \eqref{MaxCorrProof2} we obtain
\begin{align}
	C_1 + C_3 +\dots+ C_N = 2^{N-1}.
\end{align}
Theorem follows by noting that each $C_k$ is non-negative.
\end{proof}
The bound of \eqref{EQN_MAXCORR} is tight and it is achieved, e.g. by the GHZ state
\begin{align}
	\ket{GHZ}_N =\frac{1}{\sqrt 2}\left(\ket{0}^{\otimes N} + \ket{1}^{\otimes N}\right). \label{DEF_GHZ}
\end{align}
Although the theorem works only for states of odd number of qubits $N$, a similar bound, of value $2^{N-1}+1$, is observed for GHZ states of even $N$.
We conjecture that this is indeed the maximum possible value of $\mathcal{C}$ for any even $N$.
We have been uniformly sampling pure states randomly over respective spaces and so far no counter-example to the conjecture has been found.
We also note that GHZ states are not the only states with maximal value of $\mathcal{C}$.
For example, for $N = 4$ qubits the same value is attained by the double singlet state
\begin{align}
	\ket{\Psi^-}\ket{\Psi^-}=\frac{\ket{01}-\ket{10}}{\sqrt 2}\otimes \frac{\ket{01}-\ket{10}}{\sqrt 2}.
\end{align}


\subsection{Random correlations of cluster states}

As another concrete example we calculate the length of correlations for 2D cluster states~\cite{PhysRevLett.85.910,PhysRevLett.86.5188}.
Since they are universal for quantum computing their geometric measure of entanglement displays intermediate values
in agreement with findings that too highly entangled states are useless for universal quantum computing~\cite{PhysRevLett.102.190501,PhysRevLett.102.190502}.
We find the same behavior of the length of correlations.

Consider a square lattice of size $n$ with each node connected to its nearest neighbors.
At each node $a$, let there be a qubit associated with it and an operator 
\begin{align}
	K_a = \sigma^{(a)}_z\underset{b\in \mathcal{N}(a)}{\bigotimes}\sigma^{(b)}_x,
\end{align}
where the superscripts $a$ and $b$ show on which qubits the Pauli matrices act.
The tensor product is taken over the set $\mathcal{N}(a)$ of nodes neighbouring with $a$.
The cluster state $\ket{C}$ is defined as a common eigenstate of operators $K_a$~\cite{PhysRevLett.85.910}:
\begin{align}
	K_a \ket{C}=\ket{C} \:\:\:\:\:\:\:\:\:\:\:\: \forall a.
\end{align}
We compute the length of correlations for the $n \times n$ cluster states for small $n$ and extrapolate to larger $n$.
As seen in Fig.~\ref{FIG_CLUSTER_ENT}, random correlations of cluster states are halfway between product states and GHZ states, 
so that they mimic behavior of the geometric measure of entanglement.

\begin{figure}[]
\includegraphics[width=90mm]{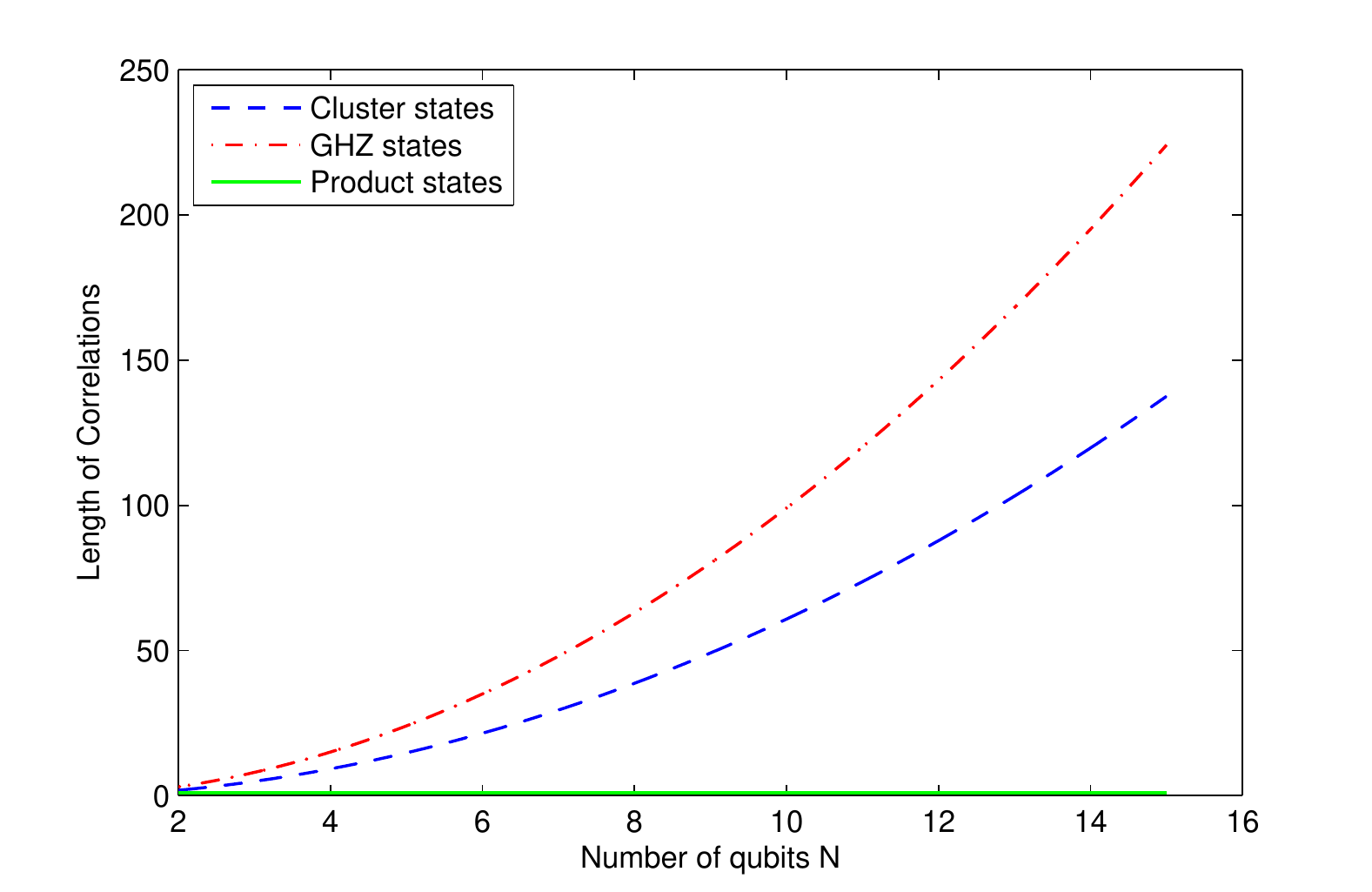}
\caption{Length of correlations for 2D cluster states (blue dashed line) and GHZ states (red dash-dot line). 
Here the length of correlations has the same features as the geometric measure of entanglement.
However, see Sec.~\ref{SEC_MONOTONE}.
The curve is obtained by fitting an exponential function to numerical data.\label{FIG_CLUSTER_ENT}}
\end{figure}


\subsection{Does $\mathcal{C}$ measure entanglement?}
\label{SEC_MONOTONE}

Since $\mathcal{C}$ and $\mathcal{R}$ perfectly distinguish pure entangled states from disentangled ones,
and in calculations of concrete examples they display proportionality to the geometric measure of entanglement, we ask if they are entanglement monotones in general.
We prove that they are the monotones for bipartite systems.
However, this does not generalize to multipartite systems as we will show on counter-examples below.

\subsubsection{Random correlations may increase on average under local operations}

Consider the following state of five qubits
\begin{align}
	\ket{\Psi} = \frac{\ket{0}\ket{GHZ}_4+\ket{1}\ket{D^2_4}}{\sqrt{2}},
\end{align}
where $\ket{GHZ}_4$ is the GHZ state of four qubits, defined in \eqref{DEF_GHZ}, and $\ket{D_4^2}$ is the four-qubit Dicke state
\begin{eqnarray}
	\ket{D^2_4} & = & \frac{1}{\sqrt{6}} (\ket{1100} + \ket{1010} + \ket{1001}, \nonumber \\
	& + & \ket{0110} + \ket{0101} + \ket{0011}).
\end{eqnarray}
It is straight forward to verify that the length of correlations of $\ket{\Psi}$ is $\mathcal{C}(\Psi) = 8$.
If a projective measurement in the computational basis is performed on the first qubit, the state will collapse to either $\ket{0}\ket{GHZ}$ or $\ket{1}\ket{D^2_4}$, both of which have length of correlations equal to $9$. 
Thus $\mathcal{C}$ increases after such local measurement independently of the actual measurement outcome. 
Thus, it is not a legitimate entanglement measure~\footnote{This section gives a counterexample to Prop. $6$ of Ref.~\cite{PhysRevA.80.042302}. We have informed the authors and they are preparing a corrigendum.}.
This five-qubit example is the simplest that we were able to find.
Therefore, in principle, the random correlations could still be entanglement monotone for systems of three and four qubits.

\subsubsection{Random correlations and LOCC conversion between pure states}

Previous section disproves strong version of monotonicity of random correlations under LOCC.
There is still a possibility that $\mathcal{R}$ is a monotone not on average, i.e. $\mathcal{R}$ could be a monotone under LOCC operations that map pure states to pure states.
We show here that this weaker form of monotonicity also does not hold for $\mathcal{R}$ when applied to multipartite systems.

For bipartite systems the following statement holds.
\begin{theorem}
For pure bipartite states, $\ket{\psi}$ can be converted by LOCC to $\ket{\phi}$ if and only if $\mathcal{C}(\psi) \ge \mathcal{C}(\phi)$.
\end{theorem}
\begin{proof}
From Nielsen's theorem~\cite{PhysRevLett.83.436}, $\ket{\psi}$ can be converted by LOCC to $\ket{\phi}$ if and only if the Schmidt probabilities $p_j(\psi)$ of $\ket{\psi}$ is majorized by $p_j(\phi)$ of $\ket{\phi}$:
\begin{equation}
\sum_{j=1}^k p_j(\psi) \le \sum_{j=1}^k p_j (\phi),
\end{equation}
for any $k=1,\dots,d$ and the Schimdt probabilities are arranged in decreasing order. 
From purity condition and the Schmidt decomposition, we find $\mathcal{C}(\psi) = d^2 + 1 - 2 d \sum_{j} p_j^2(\psi)$.
Since the square function is strictly convex in $\mathbb{R}^+$, by Karamata's inequality $p_j(\psi)$ is majorized by $p_j(\phi)$ if and only if $\sum_{j} p_j^2(\psi) \le \sum_{j} p_j^2(\phi)$, if any only if $\mathcal{C}(\psi) \ge \mathcal{C}(\phi)$.
\end{proof}

We give the following counter-example for multipartite systems.
Consider the pair of states:
\begin{eqnarray}
\ket{\psi} & = & \frac{\ket{0}\ket{\psi^-}\ket{\psi^-}+\ket{1}\ket{\psi^+}\ket{\psi^+}}{\sqrt{2}},\\
\ket{\phi} & = & \ket{0}\ket{\psi^-}\ket{\psi^-},
\end{eqnarray}
where $\ket{\psi^{\pm}} = \frac{1}{\sqrt{2}}(\ket{01} \pm \ket{10})$ are the two Bell states.
Starting with  $\ket{\psi}$ we measure the first qubit in the standard basis and, depending on the outcome, apply suitable local unitaries on say the second and fourth qubit to obtain $\ket{\phi}$. 
However, $\mathcal{C}(\psi) = 8$ whereas $\mathcal{C}(\phi) = 9$.


\subsection{Witnessing entanglement with single random setting per party}

Although random correlations do not measure entanglement we argue here that they are useful for entanglement detection.
In particular, they allow to detect quantum entanglement with high level of confidence even with a single random measurement setting per party~\cite{PhysRevA.92.050301}.
Such advantage is relevant to experiments.
For example, the small count rates in multi-photon experiment (e.g. Ref.~\cite{NatPhot.6.225}) make measurement of every next setting expensive. 
The strategy presented here reduces the number of settings required to detect entanglement to its ultimate minimum.

In principle, to determine $\mathcal{R}$, an infinite number of measurements has to be performed 
both in terms of $K$, the resources needed to estimate correlation functions, and in terms of $M$, the resources needed for averaging over random settings.
In Ref.~\cite{PhysRevA.92.050301} we introduced entanglement witness~\cite{HoroNPT,RevModPhys.81.865} that takes the finiteness of $K$ and $M$ into account,
and here we demonstrate explicitly the effect of finite $K$ on this witness.
For a single random setting, $M=1$, the witness reads:
\begin{equation}
\mathcal{R}_{K} > 1/3^N + \delta \quad \implies \quad \textrm{likely } \psi \textrm{ is ent}, \label{CRI_RANDWIT}
\end{equation}
where $1/3^N$ is the random correlation of the product state of $N$ qubits
and $\delta$ is used to set the confidence level of entanglement detection.
Namely, if estimated correlation $\mathcal{R}_{K}$ is far from what is expected for a product state, most likely we are measuring an entangled state.
In our calculations, the confidence level is defined by the probability that the random correlation of a product state is $95.4\%$.
Table~\ref{TAB_PROB}  shows the probability to detect entanglement in GHZ states (and states that can be reached from GHZ by local unitaries) of $N$ qubits with both finite and infinite $K$.
For sufficiently big $K$ the chance of detection grows with $N$.
For finite $K$ the detection probability grows for small $N$ and then starts decaying.
This is because the random correlation of any state is exponentially small in $N$, and therefore there exists critical $N$ for which the error $1/\sqrt{K}$ in estimation of the average due to finite $K$ is comparable with the bound of Eq. (\ref{CRI_RANDWIT}).
As illustration, Tab.~\ref{TAB_PROB} shows that $K=1000$ trials is essentially infinity for up to $N=5$ qubits.
For $N=6$, the bound of the entanglement witness is $\approx 0.01$ and indeed matches random correlation of the six-qubit GHZ state, $\approx 0.04$, reduced by the error $1/\sqrt{K} \approx 0.03$. 

\begin{table}[h]
\centering
\caption{Probability of detecting $N$-qubit GHZ entanglement with a single random measurement per party at confidence level of $95.4\%$.
$K$ gives the number of trials after which the correlation function is estimated.}
\label{TAB_PROB}
\begin{tabular}{r|| c c c c c c c c}
$N$                     & 3   & 4    & 5    & 6    & 7    & 8    & 9    & 10  \\ \hline 
$K=1000$                & \, 26\% & \, 44\% & \, 47\% & \, 57\% & \, 52\% & \, 48\% & \, 41\% & \, 34\%\\
$K \rightarrow \infty $ & \, 26\% & \, 44\% & \,  48\% & \, 63\% & \,  67\% & \, 77\% & \, 80\% & \, 86\%
\end{tabular}
\end{table}


\section{Mixed states}
\label{SEC_MIXED}

So far we have only considered the length of correlations in pure states.
Although the previous definition $\mathcal{C}$, as given in Eq.~(\ref{DEF_C}), suits a mixed state $\rho$ it longer identifies entanglement with certainty as it does for pure states. 
Clearly, $\mathcal{C}$ can be even less than unity for mixed states.
Nevertheless, a necessary and sufficient condition can still be established for entanglement in rank-$2$ states.

Our approach is to define a new quantity via convex roof extension of the length of correlations:
\begin{align}
	\mathcal{E}(\rho)=\min_{\{\mu_k,\Psi_k\}} \sum_k \mu_k \mathcal{C}(\Psi_k), \label{CONV_ROOF}
\end{align}
where the sum is minimized over all possible pure-state decompositions $\{\mu_k,\Psi_k\}$ of $\rho$, i.e. $\rho = \sum_k \mu_k \ket{\Psi_k}\bra{\Psi_k}$.
A state $\rho$ is entangled if and only if $\mathcal{E}(\rho)>1$.
As in other convex roof constructions, calculation of $\mathcal{E}(\rho)$ is generally a challenge. 
However, for certain families of mixed states, explicit formulae for $\mathcal{E}(\rho)$ can be found.

\subsection{Necessary and sufficient condition for entanglement in rank-two states}

Rank-two states are mixed states which belong to a subspace spanned by only two distinct pure states. 
They possess properties similar to those of a single qubit that notably simplify the minimization problem of (\ref{CONV_ROOF}). 
In what follows, we shall acquire a similar technique to that presented by Osborne \cite{PhysRevA.72.022309} to evaluate the entanglement of an arbitrary mixed state of rank two. 

\begin{theorem} \label{TH_RANK2}
For a multipartite mixed state of rank two
\begin{equation}
	\mathcal{E}(\rho) = \mathcal{C}(\rho)+\frac{1}{2} \left(1-\Tr(\rho^2) \right) w_{\min},
\end{equation}
where $\mathcal{C}(\rho)$ is given in Eq.~(\ref{DEF_C}) and $w_{\min}$ is the lowest eigenvalue of $3 \times 3$ matrix defined in the proof.
\end{theorem}
\begin{proof}
See Appendix~\ref{APP_RANK2}.
\end{proof}

The advantage of Th.~\ref{TH_RANK2} lies in its computability.
As a demonstration, we prove that a nontrivial mixture of a product state and an entangled pure state is always entangled~\cite{TheorComputSci.292.589,JPhysA.47.424025}.
It turns out that $\mathcal{E}(\rho)$ behaves similarly to entanglement quantifiers.
Since every product state can be brought to $\ket{00\dots 0}$ using a local unitary transformation, let us write such mixture in the most general form as
\begin{align}
	\rho = p \proj{00\dots 0} + (1-p) \proj{\Phi},
\end{align}
where $p$ is a probability and $\ket{\Phi}$ is a general pure state.
A compact formula for $\mathcal{E}(\rho)$ could be found if we further restrict the state $\ket{\Phi}$ to a superposition of $\ket{00\dots 0}$ and another product state $\ket{\alpha}$ orthogonal to $\ket{00\dots 0}$.
Direct application of Th.~\ref{TH_RANK2} shows
\begin{align}
	\mathcal{E}(\rho)=1+(1-p)^2 \left( \mathcal{C}(\Phi)-1 \right).
\end{align}
If $\ket{\Phi}$ is entangled, its length of correlations $\mathcal{C}(\Phi)>1$. 
In this case also $\mathcal{E}(\rho)>1$ and the mixture is entangled for all non-trivial values of $p$.

\subsection{Witness for general states}

We extend the idea used in Theorem \ref{TH_RANK2} to mixed states of arbitrary rank.
By following the same steps as in the preceding proof, with the Pauli matrices replaced by generalized Gell-Mann matrices,
we obtain a lower bound of the following theorem.
\begin{theorem} \label{TH_RANK_M}
	For a multipartite mixed state of rank $m$, 
	\begin{align}
		\mathcal{E}(\rho)\geq \mathcal{W}(\rho)\equiv \mathcal{C}(\rho)+\frac{w_{\min}}{m^2} \left(1 - \Tr(\rho^2) \right), \label{EQ_RANK_M}
	\end{align}
	where all the quantities are defined in analogy to Th.~$\ref{TH_RANK2}$.
\end{theorem}
This is no longer necessary and sufficient condition for entanglement because there might not be a physical pure state decomposition that achieves the minimum similar to equation (\ref{E}).

Nevertheless, this witness is of some interest as demonstrated by the following example where it detects all the entangled states of a certain family.
Consider three qubits in the mixed state
\begin{equation}
\rho = (1-p) \ket{W}\bra{W} + p \, \rho_{n}, \label{Rank3States}
\end{equation}
with
\begin{eqnarray}
	\ket{W} & = & \tfrac1{\sqrt3}(\ket{100}+\ket{010}+\ket{001}), \\
	\rho_{n} & = & \tfrac13(\ket{100}\bra{100}+\ket{010}\bra{010}+\ket{001}\bra{001}). \nonumber
\end{eqnarray}
It is straightforward to verify that $\rho$ is always of rank-$3$ except for trivial values of $p=0 \text{ or } 1$. 
Our witness reveals that the mixed state is entangled for all $p <1$, see Fig~\ref{FIG_WitnessVsPPT}.

\begin{figure}[t]
\includegraphics[width=90mm]{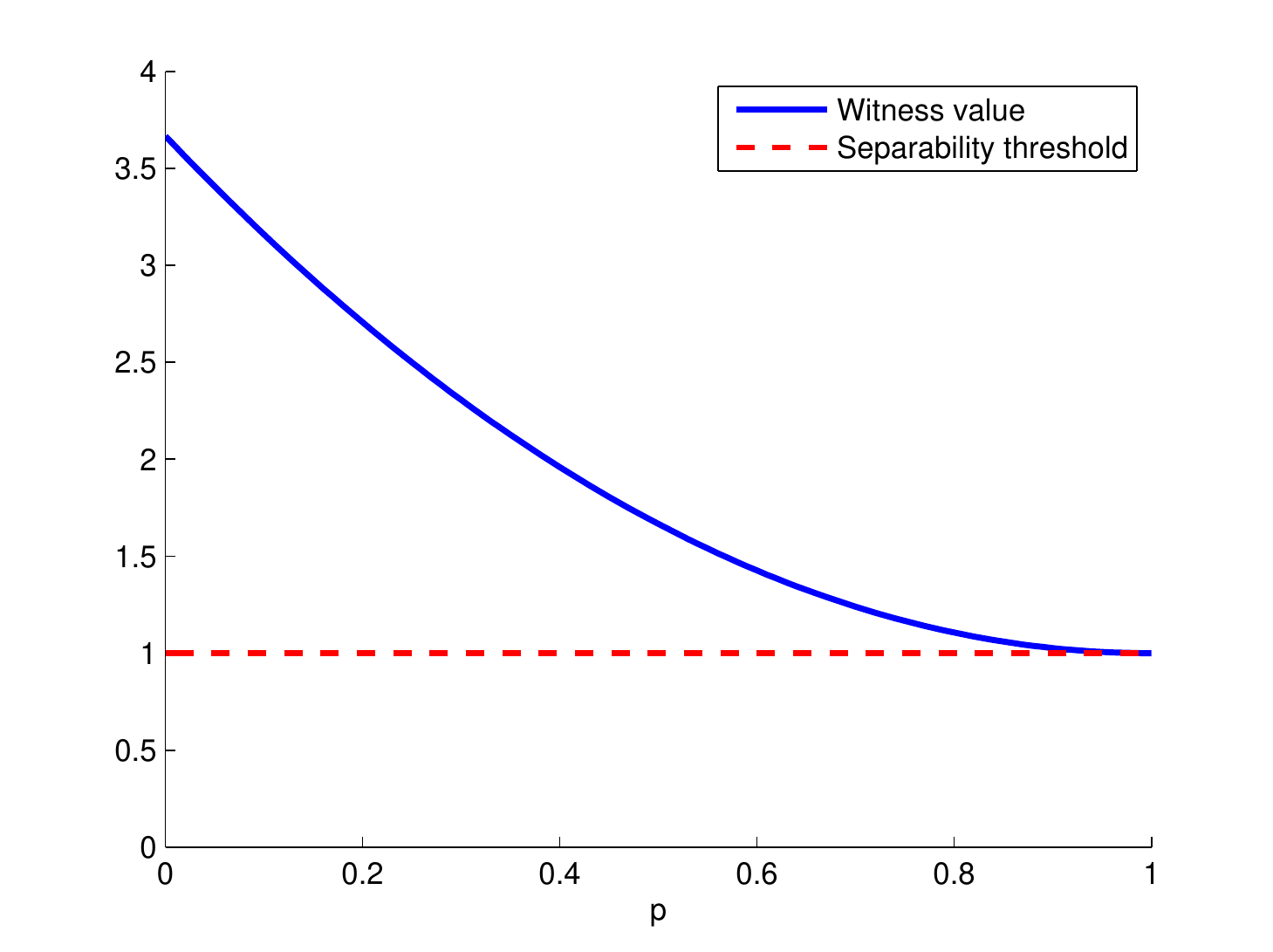}
\caption{\label{FIG_WitnessVsPPT} Witness \eqref{EQ_RANK_M} detects entanglement of all the entangled states of the family given in Eq.~\eqref{Rank3States}.}
\end{figure}


\section{Conclusion}

In conclusion, we showed that a multipartite pure quantum state of any dimensions is entangled if and only if it gives rise to higher squared correlations in random measurements.
Only correlations between all the parties are relevant.
Alternatively, the sum of all squared components of the correlation tensor is higher in all entangled pure quantum states than in product states.
Additionally to various features discussed in the main text, this provides understanding why certain pure entangled states do not violate any two-setting Bell inequalities for correlation functions~\cite{PhysRevA.64.032112,PhysRevLett.88.210401,PhysRevLett.88.210402}.
Conditions for violation of such inequalities involve a plane of correlation tensor defined by the two settings.
There exist states which have bounded correlations in every plane of the correlation tensor, but when all squared correlations are summed up the state is revealed as entangled.

\section{Acknowledgments.}
We warmly thank Micha{\l} Horodecki for stimulating discussions.
This work is supported by the National Research Foundation, Ministry of Education of
Singapore Grant No. RG98/13, start-up grant of the Nanyang Technological University,
NCN Grant No. 2012/05/E/ST2/02352, European Commission Project RAQUEL,
and Austrian Science Fund (FWF) Individual Project 2462.

\appendix


\section{Random correlations and Schmidt decomposition}
\label{APP_SCHMIDT}

\begin{theorem} \label{TH_SCHMIDT}
A pure state of two and three qubits is entangled iff its length of correlations is greater than $1$.
\end{theorem}
\begin{proof}
We will only use Schmidt decompositions and purity conditions.
For a pure state $\Psi$ of two qubits, the purity condition, $\Tr(\rho^2) = 1$, requires that 
\begin{align}
	\mathcal{C}(\Psi) + |\vec a|^2  + |\vec b|^2 = 3,
\end{align}
where $\vec a, \vec b$ are the local Bloch vectors and $\mathcal{C}(\Psi)$ is the length of correlations (see Eq.~\eqref{DEF_C}).
From Schmidt decomposition one has $|\vec a|^2  = |\vec b|^2\leq 1$.
Thus the length of correlations must satisfy $\mathcal{C}(\Psi)\geq 1$. 
The only case that $\mathcal{C}(\Psi)=1$ is when both $|\vec a|^2  = |\vec b|^2 = 1$, which means $\Psi$ is a product state.

For three qubits in a pure state $\Psi$, let us denote by $\rho_i$ the reduced state of the $i$th subsystem and by $\rho_{ij}$ the reduced state of the $i$th and $j$th subsystems together.
Schmidt decomposition requires that $\Tr(\rho^2_i) = \Tr(\rho^2_{jk})$ for every $i \neq j \neq k$.
In terms of correlations this gives
\begin{align}
	\frac{1}{2}(1+|\vec v_i|^2) = \frac{1}{4}\left(1+|\vec v_j|^2+|\vec v_k|^2+ \mathcal{C}(\rho_{jk})\right), \label{AP_SCHMIDT3}
\end{align}
where e.g. $\vec v_i$ is the Bloch vectors of $\rho_i$ and $\mathcal{C}(\rho_{jk})$ is the length of correlations of the state $\rho_{jk}$.
Note that there are three equations of the form~(\ref{AP_SCHMIDT3}).
After summing them up all the Bloch vectors cancel out and we find
\begin{align}
	\mathcal{C}(\rho_{12})+\mathcal{C}(\rho_{13})+\mathcal{C}(\rho_{23}) = 3. \label{AP_MONO3}
\end{align}
Let us recall the purity condition for $\Psi$:
\begin{align}
	|\vec v_1|^2 + |\vec v_2|^2 &+|\vec v_3|^2 + \mathcal{C}(\Psi) \nonumber \\
	&+ \mathcal{C}(\rho_{12}) +\mathcal{C}(\rho_{13})+\mathcal{C}(\rho_{23}) = 7. \label{AP_PURITY3}
\end{align}
From (\ref{AP_MONO3}) and (\ref{AP_PURITY3}) together with the fact that the length of Bloch vectors is upper-bounded by unity, we have $\mathcal{C}(\Psi) \geq 1$.
In addition, $\mathcal{C}(\Psi)$ is equal to $1$ if and only if all three local Bloch vectors are normalized, i.e. $\Psi$ is a product state. 
\end{proof}


\section{Operator bases}
\label{APP_BASES}
Two most often used operator bases that satisfy Eqs.~(\ref{COND_BASIS1}) and (\ref{COND_BASIS2}) are as follows.

The generalised Gell-Mann matrices (hermitian basis) can be divided into three classes:
\begin{eqnarray}
	G^+_{mn} & = & \sqrt{\tfrac d 2} (\ket{m}\bra{n}+\ket{n}\bra{m}),\\
	G^-_{mn} & = & \sqrt{\tfrac d 2}(-i\ket{m}\bra{n}+i\ket{n}\bra{m}),\\
	\lambda_l & = & \sqrt{\tfrac{d}{(l+1)(l+2)}}\left(\sum_{j=0}^{l} \ket{j}\bra{j}-(l+1)\ket{l+1}\bra{l+1}\right), \nonumber
\end{eqnarray}
where $0\leq m<n\leq d-1$, and $0\leq l \leq d-2$ and $d$ is the dimension of the Hilbert space of pure states. 

The Weyl-Heisenberg matrices (unitary basis) read:
\begin{eqnarray}
W_{mn} & = & X^m Z^n, \qquad m,n=0,\dots, d-1, \\
Z & = & \sum_k \exp(i 2 \pi k /d) \proj{k}, \\
X & = & \sum_k \ket{(k+1) \mod d} \bra{k}. 
\end{eqnarray}


\section{Invariance of the length of correlations}
\label{APP_CInv}

\begin{theorem}
The length of correlations, Eq. (\ref{C_D}), is invariant under the choice of local basis satisfying Eqs.~(\ref{COND_BASIS1}) and (\ref{COND_BASIS2}).
\end{theorem}
\begin{proof}
Denote by $\mathcal{C}$ the length of correlations calculated in a basis $\{\sigma_{j}\}$ for each of the $N$ qudits.
Without loss of generality, we shall prove that the same value is obtained if the local basis of the first qudit is changed to $\{\sigma'_{j}\}$.
Since only the traceless operators enter the definition of the length of correlations, i.e. $j = 1, \dots, d^2 - 1$, and both bases are complete we have
\begin{align}
	\sigma_j' = \sum_{k=1}^{d^2-1} \alpha_{jk} \sigma_k,
\end{align}
where $\alpha_{jk}$ are complex coefficients forming unitary matrix $\alpha$ because 
\begin{align}
	\sum_{k} \alpha_{jk} \alpha_{lk}^* = \frac{1}{d} \Tr \left( \sigma'_{j} (\sigma'_{l})^\dag \right) = \delta_{jl}.
	\label{EQN_ORTHOCOEFF}
\end{align}
In matrix form, $\alpha\alpha^\dag = \mathbb{I}$.
Denote by $\mathcal{C}'$ the length of correlations evaluated in the $\{\sigma'_\mu\}$ basis. 
Let $\sigma_M$ denote the tensor product of $\sigma_{j_2} \otimes \dots \otimes \sigma_{j_N}$ for the last $N-1$ qudits.
We have
\begin{align}
	\mathcal{C'} &= \sum_{M} \sum_{j} |\Tr(\rho \, \sigma'_j \otimes \sigma_M)|^2 \nonumber\\
	&=\sum_{M}\sum_{j}|\sum_k \alpha_{jk} \Tr(\rho \, \sigma_k \otimes \sigma_M)|^2 \nonumber\\
	&=\sum_{M}\sum_{j}|\sum_k \alpha_{jk} T_{k M}|^2\nonumber\\
 	&=\sum_{M}\sum_{j}\left(\sum_k |\alpha_{jk}|^2 |T_{k M}|^2 + \sum_{k \neq l} \alpha_{jk} \alpha^*_{jl} T_{k M} T_{l M}^* \right)\nonumber\\
	&=\sum_{M}\sum_{k}\left(\sum_j |\alpha_{jk}|^2\right)  |T_{k M}|^2 \nonumber \\
	&\:\:\:\:\:\:\:\:\:\:\:\:\:+ \sum_M \sum_{k \neq l} \left(\sum_j \alpha_{jk}\alpha^*_{jl} \right) T_{k M} T_{l M}^* \nonumber\\
	&=\sum_{M}\sum_{k} |T_{k M}|^2 = \mathcal{C},
\end{align}
where in the second last equality we used \eqref{EQN_ORTHOCOEFF} and $\alpha^\dag\alpha = (\alpha^{-1}\alpha)(\alpha^\dag\alpha)=\alpha^{-1}(\alpha\alpha^\dag)\alpha = \alpha^{-1}\alpha = \mathbb{I}$.
\end{proof}

\section{Length of correlations and entanglement}
\label{APP_TH_1}

Here we prove Th.~\ref{TH_1} of the main text:
\textit{Pure state $\rho = \ket{\Psi}\bra{\Psi}$ is entangled if and only if}
\begin{align}
	\mathcal{C} > (d-1)^N.
\end{align}
\begin{proof}
Clearly, if $\ket{\Psi}$ is a product state, i.e. $\ket{\Psi} = \ket{\Psi_1}\otimes\dots\otimes\ket{\Psi_N}$, then the length of correlations factors and we obtain
\begin{align}
	\mathcal{C}(\Psi) = \prod_{i=1}^N \mathcal{C}(\Psi_i)=\left(d-1\right)^N.
\end{align}
For the proof in the other direction consider two copies of the state.
In general, we can write the length of correlations as $\mathcal{C} = \Tr(\rho \otimes \rho \, \mathcal{S})$, where $\mathcal{S}$ is defined as
\begin{align}
	\mathcal{S} = S^{11'}\otimes \dots\otimes S^{NN'},
\end{align}
with $S^{nn'} = \sum_{j_n=1}^{d^2-1}\sigma_{j_n} \otimes \sigma_{j_n}$, and the superscript denotes pairs of qubits $S$ acts upon.
It can be directly verified that choosing $\sigma$'s as the Gell-Mann operators (without loss of generality as $\mathcal{C}$ is basis independent, see Appendix~\ref{APP_CInv}) $S^{nn'}$ has
\begin{enumerate}
\item $d$ eigenstates $\ket{\alpha_j}=\ket{jj}\:(0\leq j\leq d-1)$ with eigenvalue $d-1$,
\item $\frac{d(d-1)}{2}$ eigenstates $\ket{\beta_{ij}}=\ket{ij}+\ket{ji}\:(0\leq i<j\leq d-1)$ with eigenvalue $d-1$,
\item $\frac{d(d-1)}{2}$ eigenstates $\ket{\gamma_{ij}}=\ket{ij}-\ket{ji}\:(0\leq i<j\leq d-1)$ with eigenvalue $-d-1$.
\end{enumerate}
Since all the eigenvalues of $S$ are $\pm (d-1)$, all the eigenvalues of $\mathcal{S}$ are of the form $(-1)^k (d+1)^k(d-1)^{N-k}$ for $k = 0,1,...,N$.
However, since $-(d+1)$ corresponds to antisymmetric eigenstates, it must occur in pairs, i.e. $k$ must be even.
This follows from the fact that $\mathcal{S}$ is calculated for two identical copies of the state and therefore has no anti-symmetric component.
Thus the smallest eigenvalue of $\mathcal{S}$ is $(d-1)^N$.
It therefore follows that $\mathcal{C}\geq (d-1)^N$.
The equality is when $\ket{\Psi}\otimes\ket{\Psi}$ lies in the space spanned by the eigenstates $\ket{\alpha_j}$ and $\ket{\beta_{ij}}$.
Such state is symmetric with respect to the exchange of \emph{any} qudit $j$ and its copy $j'$.
Let a general expansion of $\ket{\Psi}$ in the standard basis be
\begin{align}
	\ket{\Psi}=\sum_{j_1 \dots j_N} \eta_{j_1 \dots j_N} \ket{j_1 \dots j_N},
\end{align}
where $\eta_{j_1 \dots j_N}$ are the complex coefficients.
We focus on the first two qudits and write $\ket{\Psi}=\sum\eta_{j_1 j_2 |J} \ket{j_1j_2 J}$
where $J$ stands for a sequence $j_3j_4...j_N$ for the last $N-2$ qudits.
The state of the two copies of $\ket{\Psi}$ is
\begin{align}
	\ket{\Psi}\otimes\ket{\Psi}=\sum_{j_1,j_2,J}\sum_{j_1',j_2',J'} \eta_{j_1 j_2|J}\eta_{j_1'j_2'|J'} \ket{j_1 j_2 J}\ket{j_1' j_2' J'}. \label{EQN_PSIBEFORESWAP}
\end{align}
We now exchange the first qudit:
\begin{align}
	\ket{\Psi}\otimes\ket{\Psi}=\sum_{j_1,j_2,J}\sum_{j_1',j_2',J'} \eta_{j_1'j_2|J}\eta_{j_1j_2'|J'} \ket{j_1 j_2 J}\ket{j_1'j_2'J'}. \label{EQN_PSIAFTERSWAP}
\end{align}
Comparing equations \eqref{EQN_PSIBEFORESWAP} and \eqref{EQN_PSIAFTERSWAP}, we obtain relations between the coefficients $\eta$:
\begin{align}
	\eta_{j_1 j_2|J}\eta_{j_1'j_2'|J'}=\eta_{j_1'j_2|J}\eta_{j_1j_2'|J'},
\end{align}
which holds for any $j_1,j_1',j_2,j_2',J,J'$.
In particular, for $J=J',j_1'=1$ we find
\begin{align}
	\frac{\eta_{j_1 j_2|J}}{\eta_{1j_2|J}}=\frac{\eta_{j_1j_2'|J}}{\eta_{1 j_2'|J}} \equiv k_{j_1} \text{   independent of }j_2,j_2'.
\end{align}
Using this relations we can rewrite the state $\ket{\Psi}$ as 
\begin{align}
	\ket{\Psi}&=\sum_{j_1,j_2,J} \eta_{j_1j_2|J} \ket{j_1 j_2 J}\\
	&=\sum_{j_2,J} \left( \sum_{j_1} \eta_{j_1 j_2|J} \ket{j_1}\right)\ket{j_2J}\\
	&=\sum_{j_2,J} \left( \sum_{j_1} k_{j_1} \eta_{1j_2|J} \ket{j_1}\right)\ket{j_2J}\\
	&=\left(\sum_{j_1} k_{j_1} \ket{j_1}\right) \otimes \left(\sum_{j_2,J}\eta_{1,j_2|J}\ket{j_2J}\right).
\end{align}
Thus $\ket{\Psi}$ is a tensor product of a pure state for the first qudit and another pure state for the last $N-1$ qudits.
By applying this argument iteratively we find that $\ket{\Psi}$ is fully separable. 
\end{proof}


\section{Random correlations do not depend on the initial operator in averaging}
\label{APP_INITIAL}

We present two lemmas before moving to the main theorem.
\begin{lemma}
\label{LEM_INITIAL_RHO}
For traceless and trace-orthogonal operators $\sigma_1$ and $\sigma_2$, and arbitrary state $\rho$ of two qudits
\begin{eqnarray}
	A & \equiv & \int dU \, \Tr(U \otimes U. \rho\: . U^\dag \otimes U^\dag. \sigma_1^\dag \otimes \sigma_2) \\
	& = & 0.
\end{eqnarray}
\end{lemma}
\begin{proof}
By bringing the integral inside the trace one recognizes the Werner state
\begin{equation}
\rho_W = \int dU \, U \otimes U \rho\: U^\dag \otimes U^\dag.
\end{equation}
All such states can be written in the form~\cite{Werner89}:
\begin{align}
	\rho_W = \frac{1}{d^2-d\alpha}(\mathbb{I}-\alpha P),
\end{align}
where $\alpha \in [-1,1]$ and $P$ is the swap operator 
\begin{align}
	P = \sum_{i,j} \ket{ij}\bra{ji}.
\end{align}
Since $\sigma_1^\dag \otimes\sigma_2$ is traceless, we have
\begin{align}
	A = \frac{\alpha}{d^2 - \alpha d} \Tr(P.\sigma_1^\dag \otimes \sigma_2).
\end{align}
It can be directly verified that
\begin{align}
\Tr(P.\sigma_1^\dag \otimes \sigma_2) = \Tr(\sigma_1^\dag \sigma_2).
\end{align}
The lemma follows from orthogonality of $\sigma$'s.
\end{proof}

\begin{lemma}
\label{LEM_INITIAL}
For traceless and trace-orthogonal operators $\sigma_1$ and $\sigma_2$
\begin{eqnarray}
	B & \equiv & \int dU \, \Tr(U^\dag\otimes U^\dag. \sigma_1^\dag \otimes \sigma_2 . U\otimes U) \\
	& = & 0.
\end{eqnarray}
\end{lemma}
\begin{proof}
We shall prove that $B$ has all matrix elements $\langle mn | B | kl \rangle =0$.
First write:
\begin{equation}
\langle mn | B | kl \rangle = \int dU \, \Tr(|kl \rangle \langle mn| . U^\dag\otimes U^\dag. \sigma_1^\dag \otimes \sigma_2 . U\otimes U).
\end{equation}
The diagonal elements, i.e. $k = m$ and $l = n$, vanish due to Lemma~\ref{LEM_INITIAL_RHO}, because $|kl \rangle \langle kl|$ is a valid density matrix.
For off-diagonal elements, i.e. $k \ne m$ or $l \ne n$, apply Lemma~\ref{LEM_INITIAL_RHO} to states $\rho_1 = \frac{1}{2}(\ket{kl} + \ket{mn})(\bra{kl} + \bra{mn})$ and $\rho_2 = \frac{1}{2}(\ket{kl} + i \ket{mn})(\bra{kl} -i \bra{mn})$
to obtain respectively:
\begin{eqnarray}
0 & = & \langle mn | B | kl \rangle + \langle kl | B | mn \rangle, \\
0 & = & \langle mn | B | kl \rangle - \langle kl | B | mn \rangle.
\end{eqnarray}
Sum and difference reveals that all off-diagonal elements vanish.
\end{proof}

\begin{theorem}
Random correlations
\begin{equation}
\mathcal{R}(\vec \lambda) = \frac{1}{W^N} \int d U_1 \dots \int d U_N \, \,  |\Tr(\rho \, \bigotimes_{n=1}^N U_n^\dagger \lambda_n^\dag U_n)|^2
\end{equation}
do not depend on the choice of operators $\lambda_n$ such that $\Tr(\lambda_n) = 0$ and $\Tr(\lambda_n^2) = d$.
\end{theorem}
\begin{proof}
Without loss of generality let us focus on the first subsystem and denote by $J$ the sequence of the last $N-1$ qudits.
We therefore write
\begin{equation}
\mathcal{R}(\sigma_1) \equiv \frac{1}{W^N} \int d U_J \int d U_1 \, \,  |\Tr(\rho \, U_1^\dagger \otimes U_J^\dag . \sigma_1^\dag \otimes \sigma_J^\dag . U_1 \otimes U_J)|^2.
\end{equation}
The expression involving the trace can be linearised using the second copy of the state $\rho$ as follows:
\begin{eqnarray}
&& |\Tr(\rho \, U_1^\dagger \otimes U_J^\dag . \sigma_1^\dag \otimes \sigma_J^\dag . U_1 \otimes U_J)|^2 \\
&=&
\Tr(\rho \otimes \rho . U_1^\dagger \otimes U_J^\dag \otimes U_1^\dagger \otimes U_J^\dag . \\ 
&& \sigma_1^\dag \otimes \sigma_J^\dag \otimes \sigma_1 \otimes \sigma_J . U_1 \otimes U_J \otimes U_1 \otimes U_J).
\end{eqnarray}
Without loss of generality (see Appendix~\ref{APP_CInv}) we will now consider the Weyl-Heisenberg basis.
Since all the operators within this basis are related by a unitary we have:
\begin{equation}
\mathcal{R}(\sigma_1) = \mathcal{R}(\sigma_2) = \dots = \mathcal{R}(\sigma_{d^2 - 1}).
\label{EQ_REQ}
\end{equation}
Take now operator $\lambda$ from the thesis and decompose it in this basis:
\begin{equation}
\lambda = \sum_{j = 1}^{d^2 - 1} \gamma_j \sigma_j, \qquad \textrm{with} \qquad \sum_{j=1}^{d^2-1} |\gamma_j|^2 =1,
\end{equation}
where the first equation follows from $\Tr(\lambda) = 0$, and the second from $\Tr(\lambda^2) =d$.
The random correlations calculated with $\lambda$ as the initial operator satisfy:
\begin{widetext}
\begin{eqnarray}
	\mathcal{R}(\lambda) & = & \sum_{j,k = 1}^{d^2 - 1} \gamma_j \gamma^*_k  \int dU_J \int dU_1 
	\Tr\left(\rho \otimes \rho .U_1^\dagger \otimes U_J^\dag \otimes U_1^\dagger \otimes U_J^\dag . \sigma_k^\dag \otimes \sigma_J^\dag \otimes \sigma_j \otimes \sigma_J . U_1 \otimes U_J \otimes U_1 \otimes U_J \right) \\
	& = & \sum_{j,k} \gamma_j \gamma^*_k \int dU_J \Tr \left[ \rho \otimes \rho . \left(\int dU_1 U_1^\dag \otimes U_1^\dag. \sigma_k^\dag \otimes \sigma_j .U_1 \otimes U_1 \right)
	\otimes U_J^\dag \otimes U_J^\dag . \sigma_J^\dag \otimes \sigma_J. U_J \otimes U_J \right] \\
	& = & \sum_j |\gamma_j|^2  \int dU_J \Tr \left[ \rho \otimes \rho. \left(\int dU_1 U_1^\dag \otimes U_1^\dag .\sigma_j^\dag \otimes \sigma_j . U_1 \otimes U_1 \right)
	\otimes U_J^\dag \otimes U_J^\dag. \sigma_J^\dag \otimes \sigma_J. U_J \otimes U_J \right]\\
	& = & \sum_j |\gamma_j|^2  \mathcal{R}(\sigma_j) = \mathcal{R}(\sigma_1),
\end{eqnarray}
\end{widetext}
where in the second line we isolated the first particle from the principal system and the first particle from the copy,
in the third line we use Lemma~\ref{LEM_INITIAL} and in the last line Eq.~(\ref{EQ_REQ}).
\end{proof}


\section{Theorem about the convex roof extension}
\label{APP_RANK2}

Here we prove Th.~\ref{TH_RANK2} of the main text:
\textit{For a multipartite mixed state of rank two}
\begin{equation}
	\mathcal{E}(\rho)=\mathcal{C}(\rho)+\frac{1}{2} \left(1-\Tr(\rho^2) \right) w_{\min},
\end{equation}
\textit{where $\mathcal{C}(\rho)$ is given in Eq.~(\ref{DEF_C}) and $w_{\min}$ is the lowest eigenvalue of $3 \times 3$ matrix defined in the proof below.}

\begin{proof}
By the definition of rank-two states $\rho$ can be written as a mixture of $\ket{\tilde{0}}$ and $\ket{\tilde{1}}$, which are two $N$-qubit pure states. 
Without loss of generality assume that they are mutually orthogonal.
The length of correlations of a pure state $\Psi$ can be written as expectation value of an operator $\mathcal{S}$ in the two-copy state $\ket{\Psi}\ket{\Psi}$ (see Appendix \ref{APP_TH_1}).
Therefore, we can write $\tilde{\mathcal{E}}$ for a particular decomposition $\{\mu_k,\Psi_k\}$ of $\rho$ as
\begin{align}
	\tilde{\mathcal{E}}& = \sum_k \mu_k \Tr( \mathcal{S} \, \Pi_k), \label{MIX_ENT1}
\end{align}
where we denote $\Pi_k = \proj{\Psi_k} \otimes \proj{\Psi_k}$.
Since all the pure states $\Psi_k$ are within the subspace spanned by $\ket{\tilde{0}}, \ket{\tilde{1}}$, only the projection of $\mathcal{S}$ onto this subspace will contribute to the trace in (\ref{MIX_ENT1}).
Let us therefore introduce $4 \times 4$ matrix $\tilde{\mathcal{S}}$ with matrix elements $\langle \tilde i \tilde j | \mathcal{S} | \tilde m \tilde n \rangle$, where $i,j,m,n = 0,1$.
Similarly, by introducing $4 \times 4$ matrix $\tilde{\Pi}_k$ with elements $\langle \tilde i \tilde j | \Pi_k | \tilde m \tilde n \rangle$, we can rewrite (\ref{MIX_ENT1}) as
\begin{equation}
\tilde{\mathcal{E}} = \sum_k \mu_k \Tr( \tilde{\mathcal{S}} \, \tilde \Pi_k), \label{TILDAS}
\end{equation}
We now represent operators with the tilde in terms of Pauli matrices operating in the support of $\rho$:
\begin{align}
\tilde{\Pi}_k &= \frac{1}{2} (\mathbb{I} + \vec r_k \cdot \vec \sigma) \otimes \frac{1}{2}(\mathbb{I} + \vec r_k \cdot \vec \sigma), \label{EQN_BLOCH1}\\
\tilde{\mathcal{S}}  &=\frac{1}{4}\bigg(s_0\mathbb{I}\otimes\mathbb{I}+\vec{s} \cdot \vec{\sigma}\otimes \mathbb{I}
	+\mathbb{I}\otimes\vec{s} \cdot \vec{\sigma}\nonumber\\
&\:\:\:\:\:\:\:\:\:\:\:\:\:\:\:\:\:\:\:\:\:\:\:\:\:\:\:\:\:\:\:\:\:\:\:\:\:\:\:\:\:\:+\sum_{i=1}^3 w_i\sigma_i\otimes\sigma_i\bigg), \label{EQN_BLOCH2}
\end{align}
where $s_0 = \Tr(\tilde{\mathcal{S}})$ and $\vec{\sigma}=(\sigma_1,\sigma_2,\sigma_3)$ is a vector of standard Pauli matrices chosen in such a way that  $\tilde{\mathcal{S}}$ is diagonal with real entries ordered as $w_1\geq w_2 \geq w_3$.
Note that the ``local'' parts described by vector $\vec s$ are the same since $\tilde{\mathcal{S}}$ is symmetric with respect to the exchange of the two copies. 
In this representation equation (\ref{TILDAS}) takes the form
\begin{align}
	\tilde{\mathcal{E}} = \frac{1}{4}\bigg[s_0+2 \vec{s} \cdot \vec{\rho}+\sum_k \mu_k (w_1 x_k^2+w_2 y_k^2 + w_3 z_k^2)\bigg], \label{TILDAS2}
\end{align}
where $\vec{\rho} = (\rho_x,\rho_y,\rho_z)$ is the Bloch vector representing the state $\rho$ in the subspace spanned by $\ket{\tilde{0}}$ and $\ket{\tilde{1}}$, and we have introduced components of vectors $\vec r_k = (x_k,y_k,z_k)$.
Since $\ket{\Psi_k}$ are pure states, purity condition implies $z_k^2=1-x_k^2-y_k^2$ for all $k$ in the decomposition. Equation (\ref{TILDAS2}) becomes 
\begin{align}
	\tilde{\mathcal{E}}=\frac{1}{4}\Bigg[s_0+2\vec{s} \cdot \vec{\rho}&+w_3+(w_1-w_3)\bigg(\sum_k \mu_k x_k^2\bigg)\nonumber\\
	&+(w_2-w_3)\bigg(\sum_k \mu_k y_k^2\bigg)\Bigg].
\end{align}
Note that the sums on the right hand side are quadratic and thus convex. Therefore $\sum_k \mu_k x_k^2\geq (\sum_k \mu_k x_k)^2 =  \rho_x^2$ and similarly $\sum_k \mu_k y_k^2\geq  \rho_y^2$. Both inequalities can be saturated if $x_k=\rho_x$ and $y_k=\rho_y$ for all $k$. We therefore have solved the minimization problem of the earlier convex-roof construction
\begin{align}
	\mathcal{E}(\rho)=\min_{\{\mu_k,\Psi_k\}}\tilde{\mathcal{E}} &= \frac{1}{4} [s_0+2 \vec{s} \cdot \vec{\rho}+w_3\nonumber\\
	&+(w_1-w_3)\rho_x^2+ (w_2-w_3)\rho_y^2]. \label{E}
\end{align}
Taking into account that 
\begin{align}
	\mathcal{C}(\rho)& = \frac{1}{4}[s_0+2 \vec{s} \cdot \vec{\rho}+w_1 \rho_x^2+ w_2 \rho_y^2 + w_3 \rho_z^2], \\
	\Tr(\rho^2)&=\frac{1}{2}(1+\rho_x^2 + \rho_y^2 + \rho_z^2),
\end{align}
we may simplify (\ref{E}) to
\begin{align}
	\mathcal{E}(\rho)=\mathcal{C}(\rho)+\frac{1}{2} \left(1- \Tr(\rho^2) \right) w_{\min},
\end{align}
and the theorem follows. 
\end{proof}

\bibliographystyle{apsrev4-1}
\bibliography{random_corr}

\end{document}